\documentclass[draftcls,12pt,onecolumn,oneside]{IEEEtran}
\usepackage{amsmath}
\usepackage{amssymb}
\usepackage{floatfig,epsfig}
\usepackage{subfigure}
\usepackage{verbatim}
\usepackage{pifont}
\usepackage{algorithm}
\usepackage{algorithmic}
\usepackage{alltt}
\usepackage{bm}
\usepackage{isolatin1}
\usepackage{graphics}
\usepackage{epic}
\usepackage{eepic}
\usepackage{graphpap}

\usepackage{psfrag}
\usepackage{subfigure}
\usepackage{color}

\usepackage{ifpdf}
\usepackage{cite}
\usepackage{amsmath}
\usepackage{array}
\usepackage{mdwmath}
\usepackage{mdwtab}
\usepackage{url}
\usepackage{footnote}

\newtheorem{theorem}{Theorem}

\def\be { \begin{eqnarray} }
\def\ee { \end{eqnarray} }
\ifCLASSINFOpdf
\else
\fi

\begin{document}


\title{Degrees of Freedom of Multi-hop MIMO Broadcast Networks with Delayed CSIT}

\author{Zhao~Wang, %~\IEEEmembership{Student Member, IEEE,}
        Ming~Xiao, %~\IEEEmembership{Member, IEEE,}
        Chao~Wang, %~\IEEEmembership{Member, IEEE,}
        and~Mikael~Skoglund %~\IEEEmembership{Senior Member, IEEE}% <-this % stops a space
        \thanks{The authors are with the Communication Theory Lab., School of Electrical 
         Engineering, Royal Institute of Technology (KTH), Stockholm,  
         Sweden (E-mail:\{zhaowang, mingx, chaowang, skoglund\}@kth.se).}
         }

%\markboth{Submitted to Publication}%
%{Wang, Xiao, Wang and Skoglund}

\vspace{-6ex}

\maketitle

\begin{abstract}
We study the sum degrees of freedom (DoF) of a class of
multi-layer relay-aided MIMO broadcast networks with
\emph{delayed} channel state information at transmitters (CSIT).
In the assumed network a $K$-antenna source intends to communicate to $K$
single-antenna destinations, with the help of $N-2$ layers of $K$
full-duplex single-antenna relays. We consider two practical
delayed CSIT feedback scenarios. If the source can obtain the CSI
feedback signals from all layers, we prove the optimal sum DoF of
the network to be $\frac{K}{1+\frac{1}{2}+\ldots+\frac{1}{K}}$. If
the CSI feedback is only within each hop, we show that when
$K=2$ the optimal sum DoF is $\frac{4}{3}$, and when $K\geq 3$ the
sum DoF $\frac{3}{2}$ is achievable. Our results reveal that the
sum DoF performance in the considered class of $N$-layer MIMO
broadcast networks with delayed CSIT may depend not on $N$, the
number of layers in the network, but only on $K$, the number of
antennas/terminals in each layer.
\end{abstract}

\begin{IEEEkeywords}
  Degrees of freedom, multi-hop MIMO broadcast network,
  delayed CSIT, interference alignment.
\end{IEEEkeywords}

\section{Introduction}

With the increasing interest in deploying relays in 4th
generation mobile networks, multi-user multi-hop systems
have drawn substantial research attention. In
spite of the rapid advances in the understanding of
single-hop networks, our knowledge on how to deal with inter-user
interference and design efficient transmission schemes in
multi-hop systems is relatively limited. For instance, we consider a
wireless communication system in which a $K$-antenna source
intends to communicate to $K$ single-antenna destinations. If the
source's transmission can directly reach the destinations, this
system is a well-studied $K$-user MIMO broadcast channel. It is
already known that if perfect channel state information at
transmitter (CSIT) is available, the optimal sum degrees of
freedom (DoF) of the system is $K$, while without CSIT the result
is only one. Clearly, CSIT serves as a very important factor that
influences system capacity. In practice, channel estimation is in
general performed by receivers and CSIT is typically obtained via
feedback signals sent from them. However, attaining perfect
\emph{instantaneous} CSIT in realistic systems may be a challenging task
when feedback delay is not negligible compared with channel
coherence time. To gain understanding in such scenarios,
Maddah-Ali and Tse \cite{MaddahAli2010} proposed a \emph{delayed
CSIT} concept to model the extreme case where channel coherence
time is smaller than feedback delay so that CSIT would be
completely outdated. They showed that by interference alignment (IA) 
design even the outdated CSIT can be advantageous to offer DoF 
gain achieving the optimal sum DoF of
a $K$-user MIMO broadcast channel $\frac{K}{1+\frac{1}{2}+\ldots+\frac{1}{K}}$
\cite{MaddahAli2010}. Hence, from a DoF perspective, communication
in this single-hop network is relatively well understood.

Nevertheless, if the source and the destinations are not
physically connected so that the communication has to be assisted
by intermediate relays, how many DoF are available is not clear,
especially when potentially \emph{multiple layers} of relays are
required and only delayed CSIT can be available. To study the DoF
of a multi-hop network, a straightforward \emph{cascade approach}
sees the network as a concatenation of individual single-hop
sub-networks. The network DoF is limited by the minimum DoF of all
sub-networks. In this paper, we consider a class of relay-aided
MIMO broadcast networks with a $K$-antenna source, $K$ single-antenna
destinations, and $N-2$ relay layers, each containing $K$
single-antenna full-duplex relays. Following the cascade approach,
the first hop can be treated as a $K$-user MIMO broadcast channel.
Each of the remaining hops can be seen as a $K \times K$
single-antenna X channel \cite{Cadambe2009}. Hence, the achievable
sum DoF of the considered network is $\frac{4}{3} -
\frac{2}{3(3K-1)}$, i.e. that of a $K \times K$ X channel
\cite{Abdoli2011}.

However, separating the network into individual sub-networks may
not always be a good strategy. For instance, provided perfect
instantaneous CSIT, references \cite{Jeon2009,Gou2010,Chao2011}
showed that in certain systems designing transmission by treating
all hops as a whole entity can perform strictly better than
applying the cascade approach. In this paper, we will show that
with delayed CSIT this is also the case for the considered
$N$-layer relay-aided MIMO broadcast networks. Specifically, we
focus on two delayed CSIT scenarios. In a \emph{global-range}
feedback scenario, where the CSI of all layers
can be decoded by the source, we propose a joint transmission design to
prove the optimal network sum DoF to be
$\frac{K}{1+\frac{1}{2}+\ldots+\frac{1}{K}}$. In addition, in a
\emph{one-hop-range} feedback scenario, where the CSI feedback
signals sent from each layer can only be received by its adjacent
upper-layer, we show that when $K=2$ the optimal sum DoF $\frac{4}{3}$
is achievable, and when $K \geq 3$ a DoF $\frac{3}{2}$ is
achievable. These results depend not on $N$ but only on $K$,
and are clearly better than those attained by the cascade
approach.

% System Model
\section{System Model}

\label{section:System_Model}

As shown in Fig.~\ref{Fig:N_K_Relay_BC_Def}, we consider a
multi-hop MIMO broadcast network in which a source node with $K$
transmit antennas intends to communicate to $K$ single-antenna
destinations. There is no physical link between them so that $N-2$
($N \geq 3$) layers of intermediate relay nodes, each with $K$
full-duplex single-antenna relays, are deployed to aid the
communication. The network contains a total of $N$ layers of
nodes. No connection exists between non-adjacent layers. We term
this network an $(N,K)$ relay-aided MIMO broadcast network
throughout the paper. $n_{k}$ is used to represent the node $k$
($k \in\{1,2,\ldots,K\}$) at layer-$n$ ($n \in \{2,3,
\ldots,N\}$).

\begin{figure}[t!]
\centerline{\epsfig{file=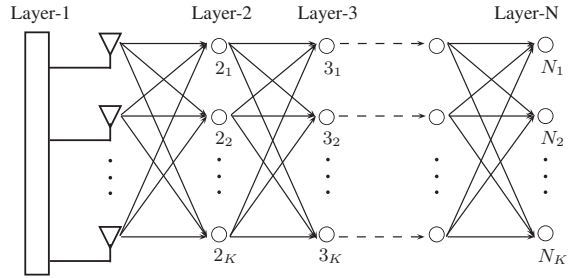,width=80mm}}
\caption{$(N,K)$ relay-aided MIMO broadcast networks.}
\label{Fig:N_K_Relay_BC_Def}
\end{figure}

% Define Dofs
Assume the rate tuple $(R_{1},R_{2},...,R_{K})$ between the source
and destinations can be achieved. Let $\mathcal{C}$ denote
the capacity region and $P$ denote the power constraint of each
layer. The sum DoF of the $(N,K)$ relay-aided MIMO broadcast
network with delayed CSIT is defined as \cite{MaddahAli2010}
\begin{equation} \label{Eqn:Eqn_DoF_Def}
D^{d-CSI}(N,K) = \max_{(R_{1},...,R_{K}) \in
\mathcal{C}}\left\{\lim_{P\rightarrow\infty}\frac{\Sigma_{i=1}^{K}R_{i}(P)}{\textrm{log}{P}}\right\}.
\end{equation}

Let a $K \times K$ matrix $\mathbf{H}^{[n-1]}(t)$ denote the
channel matrix between the $(n-1)$th and the $n$th layers (i.e.
the $(n-1)$th hop) at time slot $t$. The $i$th row
and $k$th column element of $\mathbf{H}^{[n-1]}(t)$, $h_{ik}^{[n-1]}(t)$,
represents the channel gain
from node $(n\!-\!1)_{k}$ to node $n_{i}$. We consider block
fading channels. All fading coefficients remain constant within
one time slot, but change independently across different time
slots. Let $x_{k}^{[n-1]}(t)$ ($E[|x_k^{[n-1]}(t)|^{2}] \leq \frac{P}{K}$)
and $y_{k}^{[n]}(t)$ represent the transmit signal of node
$(n\!-\!1)_{k}$ and the received signal of node $n_{k}$ at time
slot $t$, respectively. The received signals of layer-$n$ is \be
\label{Eqn:InOut_Relation_Nlayer}
   \mathbf{y}^{[n]}(t) = \mathbf{H}^{[n-1]}(t)\mathbf{x}^{[n-1]}(t) + \mathbf{z}^{[n]}(t),  n=2,3,...,N,
\ee where $\mathbf{x}^{[n-1]}(t) = [x_{1}^{[n-1]}(t)~
x_{2}^{[n-1]}(t)~ \ldots~ x_{K}^{[n-1]}(t)]^{T}$ is the transmit
signals of layer-$(n\!-\!1)$, $\mathbf{y}^{[n]}(t) =
[y_{1}^{[n]}(t)~ y_{2}^{[n]}(t)~ \ldots~ y_{K}^{[n]}(t)]^{T}$, and
$\mathbf{z}^{n}(t)$ is the unit-power complex additive white
Gaussian noise (AWGN).
%Notations $\mathbf{a}^{T}$ and $\mathbf{a}^{*}$ represent transpose
%and conjugate transpose of vector $\mathbf{a}$, respectively.

At each time slot $t$, each receiver is able to obtain the CSI of
its incoming channels by a proper training process. That is, $n_i$
knows $h_{ik}^{[n-1]}(t)$, $\forall k \in \{1,2, \ldots, K\}$.
Such knowledge can be directly delivered to nodes in later layers
along with data transmission. To transmit CSI to
previous layers, feedback signals are used from each receiver. We
assume that the feedback delay is larger than the channel coherence
time. Thus if any transmitter can receive and decode the feedback
signals, its obtained CSIT is in fact delayed by one time slot. In
this paper, we consider two scenarios of delayed CSIT feedback in
the $(N,K)$ relay-aided MIMO broadcast network:

\emph{1) Global-range delayed CSIT:} In this scenario, the source
node can receive and successfully decode the feedback signals
transmitted by all nodes. Hence it can obtain the global CSI
$\mathbf{H}^{[1]}(t), \mathbf{H}^{[2]}(t), ...,
\mathbf{H}^{[N-1]}(t)$ at time slot $t+1$.

\emph{2) One-hop-range delayed CSIT:}  In this case, the feedback
signals can be delivered only between
adjacent layers. Then at time slot $t+1$, $\mathbf{H}^{[n-1]}(t)$
is known at only layer-$(n-1)$.

% Main Results
\section{Main Results and Discussions} \label{section:Main_Results}

% Theorem state
%\subsection{Main Results} \label{Section:sub:Results_State}

We study the sum DoF of the considered $(N,K)$
relay-aided MIMO broadcast network, for both global-range and
one-hop-range delayed CSIT scenarios. Our main results are
summarized in the following two theorems.

%Associated with two feedback scenarios, we present our main results as follows.
\begin{theorem}\label{Theorem:DoF_Global_Range_Delayed_CSIT}
With \emph{global-range} delayed CSIT, the sum DoF of the $(N,K)$
relay-aided MIMO broadcast network is
\begin{equation}\label{eq:DoF_Global_Range_Delayed_CSIT}
D^{d-CSI}(N,K) = \frac{K}{1+\frac{1}{2}+\ldots+\frac{1}{K}}.
\end{equation}
\end{theorem}

\begin{proof}
Please see Section \ref{section:Proof_Theorem1} for the proof.
\end{proof}

%\vspace{1ex}
\begin{theorem}\label{Theorem:DoF_One_Hop_Range_Delayed_CSIT}
With \emph{one-hop-range} delayed CSIT, the sum DoF of the $(N,K)$
relay-aided MIMO broadcast network is
\begin{eqnarray}\label{eq:DoF_One_Hop_Range_Delayed_CSIT}
&\!\!\!\!D^{d-CSI}(N,2)\!\!\!\!& = \frac{4}{3}, \nonumber \\
\frac{3}{2} \leq &\!\!\!\!D^{d-CSI}(N,K)\!\!\!\!& \leq
\frac{K}{1+\frac{1}{2}+\ldots+\frac{1}{K}},~ K\geq3.
\end{eqnarray}
\end{theorem}

\begin{proof}
Please see Section
\ref{section:Proof_Theorem2} for the proof.
\end{proof}

%\begin{proof} See Section \ref{section:Proof_Theorem2}. \end{proof}

%\vspace{1ex}

We can see that $N$ does not appear in
(\ref{eq:DoF_Global_Range_Delayed_CSIT}) or
(\ref{eq:DoF_One_Hop_Range_Delayed_CSIT}). Thus, the sum
DoF of the $(N,K)$ relay-aided broadcast network would not be
limited by the number of layers in the network, but may be related
only to $K$, the number of antennas/users.

%\emph{\textbf{Remark}}: The achieved DoF of $(N,K)$ relay-aided broadcast network with both feedback
%are shown to rely only on $K$, and is independent to $N$. With \emph{global-range} feedback, optimal DoF
%can be achieved with any $N$ and $K$. Moreover, our achieved DoF with \emph{one-hop-range} feedback can
%achieve the upper bound when $K=2$.

% Performance evaluation

% \subsection{Discussions} \label{Section:sub:Performance_Discussion}

With \emph{global-range} feedback, the sum DoF of the network
is the same as that in a single-hop $K$-user MIMO broadcast
channel.  The result reveals the importance of providing the
CSI of the whole network to the source. In practice, this
can be achieved by e.g., each node broadcasting its feedback
signal with a sufficiently high power.
However, this may not be possible in some systems, and
\emph{one-hop-range} feedback may be more feasible. In this
case, the CSI flow is limited within only one
hop, which in turn affects the interference management in the
network. When $K\!=\!2$, the sum DoF is shown
to be $\frac{4}{3}$, by a joint transmission design among all
hops. Following a similar strategy, the sum DoF can be
lower bounded by $\frac{3}{2}$ for $K \geq 3$. Although currently
it is difficult to quantify the distance between this lower
bound and the actual achievable sum DoF, we can see that
when $K$ is small, e.g., $K=3,4$, the lower bound is tight
since it is only slightly smaller than a sum DoF upper bound.

Recall that applying the cascade approach the achievable sum DoF
is limited by that of a $K \times K$ X channel, i.e., $\frac{4}{3}
- \frac{2}{3(3K-1)}$ \cite{Abdoli2011}. By a joint transmission
design among all hops, our scheme strictly surpasses the cascade
approach. The task of proving the optimality of our results or
finding even better schemes to attain the actual sum DoF of an
$(N,K)$ relay-aided MIMO broadcast network will be left for future
investigation.

% Proof of Theorem 1
\section{Proof of Theorem \ref{Theorem:DoF_Global_Range_Delayed_CSIT}}

\label{section:Proof_Theorem1}
%
%In the proof, we first give the outer bound and then the achievable scheme.

\subsubsection{Outer Bound}
We assume that all the relays in each layer can fully
cooperate and jointly process their signals. Since this assumption
would not reduce network performance, the sum DoF of this new
system, which is clearly limited by that of the last hop (i.e. a
single-hop $K$-user MIMO broadcast channel), would serve as an
outer bound of the sum DoF of the considered $(N,K)$ relay-aided
broadcast network. According to \cite{MaddahAli2010}, the outer
bound is $\frac{K}{1+\frac{1}{2}+\ldots+\frac{1}{K}}$.

%In what follows, we will show that the above upper bound is in
%fact achievable and thus complete the proof.

% subsection: Achievability of Theorem1
\subsubsection{Achievability} Consider full-duplex
amplify-and-forward relays. At time slot $t$, node $n_{i}$ chooses
$g^{[n]}_{i}(t)$ as its amplification coefficient such that
$|g^{[n]}_{i}(t)|^{2} \left( \sum_{k=1}^{K} |h^{[n-1]}_{ik}(t)|^{2} +
\frac{1}{K} \right) \leq 1$. We define $\mathbf{G}^{[n]}(t) =
\textrm{diag}\{g^{[n]}_{1}(t),g^{[n]}_{2}(t), \ldots, g^{[n]}_{K}(t)\}$
and focus on the high-SNR regime (where DoF is effective). Hence
we omit the noise term in (\ref{Eqn:InOut_Relation_Nlayer}).
At time slot $t$, the received signals at
the layer-$N$ (i.e. the destinations) can be denoted by
\begin{equation}\label{eq:Input_Output_Relation_Global_Range_CSIT}
\mathbf{y}^{[N]}(t) =  \left( \prod_{n=3}^{N}
    \mathbf{H}^{[n-1]}(t) \mathbf{G}^{[n-1]}(t) \right) \mathbf{H}^{[1]}(t)
    \mathbf{x}^{[1]}(t).
\end{equation}

Let $\tilde{\mathbf{H}}(t) = \prod_{n=3}^{N} \left(
\mathbf{H}^{[n-1]}(t) \mathbf{G}^{[n-1]}(t) \right)
\mathbf{H}^{[1]}(t)$ and substitute it into
(\ref{eq:Input_Output_Relation_Global_Range_CSIT}). We obtain an
equivalent single-hop $K$-user MIMO broadcast channel with
an equivalent channel matrix $\tilde{\mathbf{H}}(t)$. Because the
source has the delayed CSI of the whole network,
$\tilde{\mathbf{H}}(t-1)$ is known at time slot $t$. Thus the
transmission scheme proposed in \cite{MaddahAli2010} for achieving
the sum DoF of a single-hop $K$-user MIMO broadcast channel with
delayed CSIT can also be employed in the equivalent system.
The sum DoF outer bound
$\frac{K}{1+\frac{1}{2}+\ldots+\frac{1}{K}}$ is achievable.

% Proof of Theorem 2
\section{Proof of Theorem $2$}

\label{section:Proof_Theorem2}

Clearly, the outer bound above still holds for
\emph{one-hop-range} feedback. When $K=2$ it can be shown
that the outer bound $\frac{4}{3}$ is tight. However, 
it may not be true for $K \geq 3$. In this proof we will present a new
multi-round transmission scheme that treats all hops as a whole
entity aiming for aligning interference. The achievable sum DoF is 
higher than that obtained by the cascade approach and thus will serve 
as a lower bound to the sum DoF of the considered network.

Due to the page limitation, we will mainly focus on an example
$(3,3)$ relay-aided MIMO broadcast network. Let integer $l \geq
1$. We will show that $9l$ independent messages can be delivered
from the $3$-antenna source to the $3$ single-antenna destinations
through a layer of $3$ single-antenna full-duplex relays, using a
total of $6l+3$ time slots. Then when $l \rightarrow \infty$, the
sum DoF $\frac{3}{2}$ can be asymptotically achieved. The
corresponding approach for general networks will be given later.

\begin{table*}
\addtolength{\tabcolsep}{-2pt} \centering \caption{}
%The block
%network coding scheme for the $(3,3)$ relay-aided MIMO broadcast
%network} \scriptsize
    \begin{tabular}{c | c c c | c c c | c c c | c c c}
        \hline %\\ [0.5pt]
        $t$             & $1$               & $2$               & $3$              & $4$                    & $5$                    & $6$
            & $7$                    & $8$                     & $9$                    & $10$                   & $11$                   & $12$                   \\ \hline
        $x_{1}^{[1]}(t)$   & $\mu_{1}(1)$      & $\mu_{2}(1)$      & $\mu_{3}(1)$     & $L_{2}^{[2]}(1)$       & $L_{3}^{[2]}(1)$       &
            & $\mu_{1}(2)$           & $\mu_{2}(2)$            & $\mu_{3}(2)$           & $L_{2}^{[2]}(2)$       & $L_{3}^{[2]}(2)$       &                        \\
        $x_{2}^{[1]}(t)$   & $\nu_{1}(1)$      & $\nu_{2}(1)$      & $\nu_{3}(1)$     & $L_{4}^{[2]}(1)$       &                        & $L_{6}^{[2]}(1)$
            & $\nu_{1}(2)$           & $\nu_{2}(2)$            & $\nu_{3}(2)$           & $L_{4}^{[2]}(2)$       &                        & $L_{6}^{[2]}(2)$       \\
        $x_{3}^{[1]}(t)$   & $\omega_{1}(1)$   & $\omega_{2}(1)$   & $\omega_{3}(1)$  &                        & $L_{7}^{[2]}(1)$       & $L_{8}^{[2]}(1)$
            & $\omega_{1}(2)$        & $\omega_{2}(2)$         & $\omega_{3}(2)$        &                        & $L_{7}^{[2]}(2)$       & $L_{8}^{[2]}(2)$       \\ \hline
        $y_{1}^{[2]}(t)$   & $L_{1}^{[2]}(1)$  & $L_{4}^{[2]}(1)$  & $L_{7}^{[2]}(1)$ & $\gamma_{12}^{[2]}(1)$ & $\gamma_{13}^{[2]}(1)$ &
            & $L_{1}^{[2]}(2)$       & $L_{4}^{[2]}(2)$        & $L_{7}^{[2]}(2)$       & $\gamma_{12}^{[2]}(2)$ & $\gamma_{13}^{[2]}(2)$ &                      \\
        $y_{2}^{[2]}(t)$   & $L_{2}^{[2]}(1)$  & $L_{5}^{[2]}(1)$  & $L_{8}^{[2]}(1)$ & $\gamma_{12}^{[2]}(1)$ &                        & $\gamma_{23}^{[2]}(1)$
            & $L_{2}^{[2]}(2)$       & $L_{5}^{[2]}(2)$        & $L_{8}^{[2]}(2)$       & $\gamma_{12}^{[2]}(2)$ &                        & $\gamma_{23}^{[2]}(2)$ \\
        $y_{3}^{[2]}(t)$   & $L_{3}^{[2]}(1)$  & $L_{6}^{[2]}(1)$  & $L_{9}^{[2]}(1)$ &                        & $\gamma_{13}^{[2]}(1)$ & $\gamma_{23}^{[2]}(1)$
            & $L_{3}^{[2]}(2)$       & $L_{6}^{[2]}(2)$        & $L_{9}^{[2]}(2)$       &                        & $\gamma_{13}^{[2]}(2)$ & $\gamma_{23}^{[2]}(2)$ \\ \hline
        $x_{1}^{[2]}(t)$   &                   &                   &                  & $L_{1}^{[2]}(1)$       & $L_{4}^{[2]}(1)$       & $L_{7}^{[2]}(1)$
            & $L_{2}^{[3]}(1)$       & $L_{3}^{[3]}(1)$        &                        & $L_{1}^{[2]}(2)$       & $L_{4}^{[2]}(2)$       & $L_{7}^{[2]}(2)$       \\
        $x_{2}^{[2]}(t)$   &                   &                   &                  & $L_{2}^{[2]}(1)$       & $L_{5}^{[2]}(1)$       & $L_{8}^{[2]}(1)$
            & $L_{4}^{[3]}(1)$       &                         & $L_{6}^{[3]}(1)$       & $L_{2}^{[2]}(2)$       & $L_{5}^{[2]}(2)$       & $L_{8}^{[2]}(2)$       \\
        $x_{3}^{[2]}(t)$   &                   &                   &                  & $L_{3}^{[2]}(1)$       & $L_{6}^{[2]}(1)$       & $L_{9}^{[2]}(1)$
            &                        & $L_{7}^{[3]}(1)$        & $L_{8}^{[3]}(1)$       & $L_{3}^{[2]}(2)$       & $L_{6}^{[2]}(2)$       & $L_{9}^{[2]}(2)$       \\ \hline
        $y_{1}^{[3]}(t)$   &                   &                   &                  & $L_{1}^{[3]}(1)$       & $L_{4}^{[3]}(1)$       & $L_{7}^{[3]}(1)$
            & $\gamma_{12}^{[3]}(1)$ & $\gamma_{13}^{[3]}(1)$  &                        & $L_{1}^{[3]}(2)$       & $L_{4}^{[3]}(2)$       & $L_{7}^{[3]}(2)$       \\
        $y_{2}^{[3]}(t)$   &                   &                   &                  & $L_{2}^{[3]}(1)$       & $L_{5}^{[3]}(1)$       & $L_{8}^{[3]}(1)$
            & $\gamma_{12}^{[3]}(1)$ &                         & $\gamma_{23}^{[3]}(1)$ & $L_{2}^{[3]}(2)$       & $L_{5}^{[3]}(2)$       & $L_{8}^{[3]}(2)$       \\
        $y_{3}^{[3]}(t)$   &                   &                   &                  & $L_{3}^{[3]}(1)$       & $L_{6}^{[3]}(1)$       & $L_{9}^{[3]}(1)$
            &                        & $\gamma_{13}^{[3]}(1)$  & $\gamma_{23}^{[3]}(1)$ & $L_{3}^{[3]}(2)$       & $L_{6}^{[3]}(2)$       & $L_{9}^{[3]}(2)$       \\ \hline
    \end{tabular}
\label{Table:BlockCoding}
\end{table*}

Recall that we use $y_k^{[n]}(t)$ and $x_k^{[n]}(t)$ respectively to
denote the received and transmitted signals of the
$k$th node in layer-$n$ (or the $k$th antenna if $n=1$) at time
slot $t$. The transmission process in the $(3,3)$ relay-aided MIMO
broadcast network, for the first $12$ time slots, is shown in
Table \ref{Table:BlockCoding}. Specifically, $2$ rounds of
messages, each containing $9$ independent messages, are delivered
to the destinations. Let $\mu_k(l)$, $\nu_k(l)$, and $\omega_k(l)$
($k\in \{1,2,3\}$) denote the source messages intended for the
destinations $3_k$ (the index $l$ means that the notations apply
for the $l$th transmission round). In what follows, we will
explain the first round of transmission. It consists of two
\emph{phases}.

\textbf{Phase One:} The first phase takes the first $3$ time
slots. At time slot $t$ ($t \in \{1,2,3\}$), $\mu_{t}(1),
\nu_{t}(1), \omega_{t}(1)$ are transmitted by the three source
antennas respectively. Hence each relay (i.e. each node of
layer-$2$) receives a linear combination of three messages at each
time slot. Again, we ignore the noise in
(\ref{Eqn:InOut_Relation_Nlayer}). The received signals at $2_k$
is expressed as ($t \in \{1,2,3\}$)
\begin{eqnarray}
y_k^{[2]}(t) = h_{k1}^{[1]}(t)\mu_{t}(1) + h_{k2}^{[1]}(t)\nu_{t}(1)
    + h_{k3}^{[1]}(t)\omega_{t}(1).
\end{eqnarray}
Let $L_{3(t-1)+k}^{[2]}(1)=y_k^{[2]}(t)$ denote the linear
equation known by $2_k$ at time slot $t$.
%We can express these
%linear combinations as
%\begin{eqnarray} \label{Eqn:Eqn_NlayerRepre1}
%    L_{3(t-1)+1}^{[2]}(1) \!=\! h_{11}^{[1]}(t)\mu_{t}(1) \!+\!
%            h_{12}^{[1]}(t)\nu_{t}(1) \!+\! h_{13}^{[1]}(t)\omega_{t}(1); \\
%    L_{3(t-1)+2}^{[2]}(1) \!=\! h_{21}^{[1]}(t)\mu_{t}(1) \!+\!
%            h_{22}^{[1]}(t)\nu_{t}(1) \!+\! h_{23}^{[1]}(t)\omega_{t}(1); \\
%    L_{3(t-1)+3}^{[2]}(1) \!=\! h_{31}^{[1]}(t)\mu_{t}(1) \!+\!
%            h_{32}^{[1]}(t)\nu_{t}(1) \!+\! h_{33}^{[1]}(t)\omega_{t}(1).
%\end{eqnarray}
After the $3$rd time slot, since $\mathbf{H}^{[1]}(1)$,
$\mathbf{H}^{[1]}(2)$, and $\mathbf{H}^{[1]}(3)$ are known at the
source, all the equations $L_{i}^{[2]}(1)$, $\forall
i=1,2,\cdots,9$, can be recovered by the source.

\textbf{Phase Two:} This phase takes the next $6$ time slots after
phase one. At each time slot $t$ ($t \in \{3,4,5\}$) only two
source antennas are activated to retransmit the equations
$L_{i}^{[2]}(1)$. According to $x_{k}^{[1]}(t)$
shown in Table \ref{Table:BlockCoding}, we have
\begin{align}
    \label{Eqn:Eqn_OrdertwoSymbol}
        y_{k}^{[2]}(4) = h_{k1}^{[1]}(4) L_{2}^{[2]}(1) + h_{k2}^{[1]}(4) L_{4}^{[2]}(1); \\
        y_{k}^{[2]}(5) = h_{k1}^{[1]}(4) L_{3}^{[2]}(1) + h_{k3}^{[1]}(4) L_{7}^{[2]}(1); \\
        y_{k}^{[2]}(6) = h_{k2}^{[1]}(4) L_{6}^{[2]}(1) + h_{k3}^{[1]}(4) L_{8}^{[2]}(1).
\end{align}
Since node $2_1$ obtains $L_{4}^{[2]}(1)$ in phase one, at time
slot $4$ it can recover $L_{2}^{[2]}(1)$. Similarly, both
$L_{2}^{[2]}(1)$ and $L_{4}^{[2]}(1)$ are also known at node
$2_2$. Use $\gamma_{ij}^{[n]}(l)=(a,b)$ to represent that
equations $a$ and $b$ are recovered by both nodes $n_i$ and $n_j$.
As shown in Table \ref{Table:BlockCoding}, we can replace both
$y_{1}^{[2]}(4)$ and $y_{2}^{[2]}(4)$ with $\gamma_{12}^{[2]}(1) =
(L_{2}^{[2]}(1),L_{4}^{[2]}(1))$. Clearly, we also have
$\gamma_{13}^{[2]}(1) = (L_{3}^{[2]}(1), L_{7}^{[2]}(1))$ and
$\gamma_{23}^{[2]}(1) = (L_{6}^{[2]}(1),L_{8}^{[2]}(1))$.

Meanwhile, the relay nodes also send the equations they received
in phase one to the destinations, as shown in Table
\ref{Table:BlockCoding}. The received equations at the
destinations $3_{k}$ are:
\begin{eqnarray}
 \label{Eqn:Eqn_NlayerRepre2_1}
 \!\!\!\!\!y_{k}^{[3]}(4) \!\!\!\!\!&=&\!\!\!\!\! h_{k1}^{[2]}(4)L_{1}^{[2]}(1) \!+\!
            h_{k2}^{[2]}(4)L_{2}^{[2]}(1) \!+\! h_{k3}^{[2]}(4)L_{3}^{[2]}(1), \\
 \label{Eqn:Eqn_NlayerRepre2_2}
 \!\!\!\!\!y_{k}^{[3]}(5) \!\!\!\!\!&=&\!\!\!\!\! h_{k1}^{[2]}(5)L_{4}^{[2]}(1) \!+\!
            h_{k2}^{[2]}(5)L_{5}^{[2]}(1) \!+\! h_{k3}^{[2]}(5)L_{6}^{[2]}(1), \\
    \label{Eqn:Eqn_NlayerRepre2_3}
 \!\!\!\!\!y_{k}^{[3]}(6) \!\!\!\!\!&=&\!\!\!\!\! h_{k1}^{[2]}(6)L_{7}^{[2]}(1) \!+\!
            h_{k2}^{[2]}(6)L_{8}^{[2]}(1) \!+\! h_{k3}^{[2]}(6)L_{9}^{[2]}(1).
\end{eqnarray}

%\begin{align}
%    \label{Eqn:Eqn_NlayerRepre2_1}
%        L_{k}^{[3]}(1) \!=\! h_{k1}^{[2]}(4)L_{1}^{[2]}(1) \!+\!
%            h_{k2}^{[2]}(4)L_{2}^{[2]}(1) \!+\! h_{k3}^{[2]}(4)L_{3}^{[2]}(1); \\
%    \label{Eqn:Eqn_NlayerRepre2_2}
%        L_{3\!+\!k}^{[3]}(1) \!=\! h_{k1}^{[2]}(5)L_{4}^{[2]}(1) \!+\!
%            h_{k2}^{[2]}(5)L_{5}^{[2]}(1) \!+\! h_{k3}^{[2]}(5)L_{6}^{[2]}(1); \\
%    \label{Eqn:Eqn_NlayerRepre2_3}
%        L_{6\!+\!k}^{[3]}(1) \!=\! h_{k1}^{[2]}(6)L_{7}^{[2]}(1) \!+\!
%            h_{k2}^{[2]}(6)L_{8}^{[2]}(1) \!+\! h_{k3}^{[2]}(6)L_{9}^{[2]}(1). \end{align}

Let $L_{3(t-4)+k}^{[3]}(1)=y_k^{[3]}(t)$. Clearly, if the
destination $3_1$ knows the three equations $L_{1}^{[3]}(1)$,
$L_{2}^{[3]}(1)$, $L_{3}^{[3]}(1)$, it can recover its desired
messages $\mu_1(1)$, $\nu_1(1)$, $\omega_1(1)$. After time slot
$6$, the node $3_1$ has $L_{1}^{[3]}(1)$. Thus if $L_{2}^{[3]}(1)$
and $L_{3}^{[3]}(1)$ can be provided to node $3_1$, the problem is
solved. Similarly, having $L_{5}^{[3]}(1)$, the destination
$3_{2}$ needs $L_{4}^{[3]}(1)$ and $L_{6}^{[3]}(1)$ to recover
$\mu_{2}(1), \nu_{2}(1), \omega_{2}(1)$. $L_{7}^{[3]}(1)$ and
$L_{8}^{[3]}(1)$ are desired by the destination $3_{3}$, who
already has $L_{9}^{[3]}(1)$, to recover $\mu_{3}(1)$,
$\nu_{3}(1)$, $\omega_{3}(1)$. Therefore, we aim to deliver these
six equations from the relays to the destinations in the next
three time slots.

According to the above description, we can see that after time
slot $6$, node $2_{1}$ knows the equations $L_{1}^{[2]}(1)$,
$L_{2}^{[2]}(1)$ and $L_{3}^{[2]}(1)$. Node $2_2$ knows the
equations $L_{4}^{[2]}(1)$, $L_{5}^{[2]}(1)$ and $L_{6}^{[2]}(1)$.
Node $2_3$ knows the equations $L_{7}^{[2]}(1)$,
$L_{8}^{[2]}(1)$ and $L_{9}^{[2]}(1)$. Since the channel matrices
$\mathbf{H}^{[2]}(4)$, $\mathbf{H}^{[2]}(5)$, and
$\mathbf{H}^{[2]}(6)$ are available at all nodes in layer-$2$,
the node $2_1$ can formulate the equations
$L_{2}^{[3]}(1)$ and $L_{3}^{[3]}(1)$ using
(\ref{Eqn:Eqn_NlayerRepre2_1}). Similarly, the node $2_2$ can
formulate the equations $L_{4}^{[3]}(1)$ and $L_{6}^{[3]}(1)$
according to (\ref{Eqn:Eqn_NlayerRepre2_2}). The node $2_{3}$
can formulate the equations $L_{7}^{[3]}(1)$ and $L_{8}^{[3]}(1)$
using (\ref{Eqn:Eqn_NlayerRepre2_3}).

At time slot $7$, let $2_1$ transmit $L_{2}^{[3]}(1)$ and $2_2$
transmit $L_{4}^{[3]}(1)$, as shown in Table
\ref{Table:BlockCoding}. Node $3_1$, which already knows
$L_{4}^{[3]}(1)$, can recover $L_{2}^{[3]}(1)$ by eliminating
$L_{4}^{[3]}(1)$ from its received signal. The node $3_2$ can also
attain both $L_{2}^{[3]}(1)$ and $L_{4}^{[3]}(1)$, following the
similar approach. Thus the received signals $y_1^{[3]}(7)$ and
$y_2^{[3]}(7)$ in Table \ref{Table:BlockCoding} can be replaced
with a simpler expression $\gamma_{12}^{[3]}(1) =
(L_{2}^{[3]}(1),L_{4}^{[3]}(1))$. Then we can also have
$\gamma_{13}^{[3]}(1) = (L_{3}^{[3]}(1), L_{7}^{[3]}(1))$ and
$\gamma_{23}^{[3]}(1) = (L_{6}^{[3]}(1),L_{8}^{[3]}(1))$, at the
$8$th and $9$th time slots, respectively.

Consequently, equations $L_{1}^{[3]}(1)$, $L_{2}^{[3]}(1)$,
and $L_{3}^{[3]}(1)$ are known at the destination $3_1$. The
desired messages can be recovered now. The same result holds also
for the destinations $3_2$ and $3_3$. $9$ independent
messages are delivered successfully from the source to the
destinations in one transmission round.
%In addition, the source can start the transmission of the second
%round of $9$ messages from time slot $7$, following the above
%approach as shown in Table \ref{Table:BlockCoding}.
The same process can continue until $l$ rounds of transmissions
are finished using a total of $6l+3$ time slots (the second round
transmission is shown in Table \ref{Table:BlockCoding} partially).
When $l \rightarrow \infty$, this scheme achieves a sum
DoF $\frac{3}{2}$. The lower bound for $D^{d-CSI}(3,3)$ is proven.

To generalize this scheme to $N$ ($N\!>\!3$) layers, we first denote
the messages from the source as:
$L^{[1]}_{3(k-1)+1}(l)=\mu_{k}(l) $, $L^{[1]}_{3(k-1)+2}(l)=\nu_{k}(l)$
and $L^{[1]}_{3(k-1)+3}(l)=\omega_{k}(l)$. The
$l$th-round transmission at layer-$n$ ($n \in \{1,2,\ldots,N-1\}$)
can be denoted by the following formula. It takes the time slots
$t=6(l\!-\!1)\!+\!3(n\!-\!1)\!+\!\hat{t}$ ($\hat{t}=1,2,\ldots,6$):
\begin{equation} \label{Eqn:GeneralizedFrom}
     \mathbf{x}^{[n]}(t) =
        \begin{cases}
            \{L^{[n]}_{3(\hat{t}-1)+k}(l)\}_{k=1}^{3} & \quad \hat{t} = 1,2,3; \\
            [L^{[n+1]}_{2}(l), L^{[n+1]}_{4}(l), 0]^{T} & \quad \hat{t} = 4; \\
            [L^{[n+1]}_{3}(l), 0, L^{[n+1]}_{7}(l)]^{T} & \quad \hat{t} = 5; \\
            [0, L^{[n+1]}_{6}(l), L^{[n+1]}_{8}(l)]^{T} & \quad \hat{t} = 6. \\
        \end{cases}
\end{equation}
Here $\{L^{[n]}_{3(\hat{t}-1)+k}(l)\}_{i=k}^{3}$ represents the column
vector composed by $L^{[n]}_{3(\hat{t}-1)+k}(l)$ ($k \in \{1,2,3\}$).
We denote the received equation at node $(n\!+\!1)_{k}$ when
$\hat{t} \in \{1,2,3\}$ as $L_{3(\hat{t}-1)+k}^{[n+1]}(l)=
\sum_{i=1}^{3} h_{ki}^{[n]}(t)$ $L_{3(\hat{t}-1)+i}^{[n]}(l)$. By
induction, we assume $n_{k}$ can recover $L_{3(k-1)+i}^{[n+1]}(l)$
after the first three time slots ($i \in \{1,2,3\}$). Then the
transmission can be designed as shown in (\ref{Eqn:GeneralizedFrom}) when
$\hat{t} \in \{4,5,6\}$. Therefore, $(n\!+\!1)_{1}$ and $(n\!+\!1)_{2}$
can recover $\gamma_{12}^{[n+1]}(l) \!=\! (L_{2}^{[n+1]}(l),
L_{4}^{[n+1]}(l))$; $(n\!+\!1)_1$ and $(n\!+\!1)_3$ can recover
$\gamma_{13}^{[n+1]}(l)\!=\!(L_{3}^{[n+1]}(l),L_{7}^{[n+1]}(l))$;
and $(n\!+\!1)_2$ and $(n\!+\!1)_3$ can recover $\gamma_{23}^{[n+1]}(l) =
(L_{6}^{[n+1]}(l),L_{8}^{[n+1]}(l))$ after time slot $6l+3(n-1)$.
Since the destinations refer to the $N$th layer,
the $l$-round transmission takes $6l+3(N-2)$ time slots 
to deliver $9l$ independent messages. The achievable
sum DoF is $\frac{9l}{6l+3(N-2)} \approx \frac{3}{2}$ when
$l \rightarrow \infty$. The result still holds for $K>3$.

Now we consider $K=2$. In this case, $4$ messages are delivered
using $3$ time slots. Let $L^{[1]}_{1}=\mu_{1}$ and
$L^{[1]}_{2}=\nu_{1}$ denote the messages for the first
destination, and $L^{[1]}_{3}=\mu_{2}$ and $L^{[1]}_{4}=\nu_{2}$
denote those for the second destination. During time slot $1$, the
nodes (or antennas) $n_1$ and $n_2$ ($n \in \{1,2,\cdots,N\}$)
send $L^{[n]}_{1}$ and $L^{[n]}_{2}$, respectively. The received
signal at node $(n+1)_k$ ($k \in \{1,2\}$) is $L^{[n+1]}_{k} =
h_{k1}^{[n]}(1)L_1^{[n]}+ h_{k2}^{[n]}(1)L_{2}^{[n]}$. During time
slot $2$, $n_1$ sends $L^{[n]}_{3}$ and $n_2$ sends $L^{[n]}_{4}$.
$(n+1)_k$ receives $L^{[n+1]}_{k+2} = h_{k1}^{[n]}(2)L_3^{[n]}+
h_{k2}^{[n]}(2)L_{4}^{[n]}$. Assume $n_1$ can recover
$L^{[n+1]}_{2}$, and $n_2$ can recover $L^{[n+1]}_{3}$. During
time slot $3$, $n_1$ and $n_2$ transmit $L^{[n+1]}_{2}$ and
$L^{[n+1]}_{3}$, respectively. Since $(n+1)_1$ knows
$L^{[n+1]}_{3}$, it can recover $L^{[n+1]}_{2}$. Similarly,
$(n+1)_1$ can recover $L^{[n+1]}_{3}$. As a result, the
destination $N_1$ can thus obtain both $\mu_{1}$ and $\nu_{1}$
because it can have $L^{[N]}_{1}$ and $L^{[N]}_{2}$. The
destination $N_2$ can obtain $\mu_{2}$ and $\nu_{2}$ from
$L^{[N]}_{3}$ and $L^{[N]}_{4}$. The achieved sum DoF is
$\frac{4}{3}$ to meet the upper bound.

\section{Conclusions}
We investigate the sum DoF of a class of multi-hop MIMO
broadcast network with delayed CSIT feedback. Our results 
show the transmission design by treating the
multi-hop network as an entity can achieve better sum DoF
than the cascade approach which separates each hop individually.

\bibliographystyle{IEEEtran}
\bibliography{IEEEabrv,MyLibraryTsp}

\begin{thebibliography}{1}
\providecommand{\url}[1]{#1}
\csname url@samestyle\endcsname
\providecommand{\newblock}{\relax}
\providecommand{\bibinfo}[2]{#2}
\providecommand{\BIBentrySTDinterwordspacing}{\spaceskip=0pt\relax}
\providecommand{\BIBentryALTinterwordstretchfactor}{4}
\providecommand{\BIBentryALTinterwordspacing}{\spaceskip=\fontdimen2\font plus
\BIBentryALTinterwordstretchfactor\fontdimen3\font minus
  \fontdimen4\font\relax}
\providecommand{\BIBforeignlanguage}[2]{{%
\expandafter\ifx\csname l@#1\endcsname\relax
\typeout{** WARNING: IEEEtran.bst: No hyphenation pattern has been}%
\typeout{** loaded for the language `#1'. Using the pattern for}%
\typeout{** the default language instead.}%
\else
\language=\csname l@#1\endcsname
\fi
#2}}
\providecommand{\BIBdecl}{\relax}
\BIBdecl

\bibitem{MaddahAli2010}
M.~A. Maddah-Ali and D.~Tse, ``Completely stale transmitter channel state
  information is still very useful,'' in \emph{2010 Allerton Conference.}

\bibitem{Cadambe2009}
V.~R. Cadambe and S.~A. Jafar, ``Interference alignment and the degrees of
  freedom of wireless {X} networks,'' \emph{{IEEE} Trans. Inform. Theory},
  vol.~55, pp. 3893--3908, Sep. 2009.

\bibitem{Abdoli2011}
\BIBentryALTinterwordspacing
M.~J. Abdoli, A.~Ghasemi, and A.~K. Khandani, ``On the degrees of freedom of
  {K}-user {SISO} interference and {X} channels with delayed {CSIT}.''
  [Online]. Available: \url{http://arxiv.org/pdf/1109.4314}
\BIBentrySTDinterwordspacing

\bibitem{Jeon2009}
S.~W. Jeon, S.~Y. Chung, and S.~A. Jafar, ``Degrees of freedom of multi-source
  relay networks,'' in \emph{2009 Allerton Conference}.

\bibitem{Gou2010}
T.~Gou, S.~A. Jafar, S.~W. Jeon, and S.~Y. Chung, ``Aligned interference
  neutralization and the degrees of freedom of the 2x2x2 interference
  channel,'' \emph{{IEEE} Trans. Inform. Theory}, vol.~58, pp. 4381--4395, Jul.
  2012.

\bibitem{Chao2011}
C.~Wang, H.~Farhadi, and M.~Skoglund, ``Achieving the degrees of freedom of
  wireless multi-user relay networks,'' \emph{{IEEE} Trans. Commun.}, To
  appear.

\end{thebibliography}

\end{document}